\documentclass[a4paper,UKenglish]{lipics}
 
\usepackage{microtype}

\usepackage{amsthm}
\usepackage{amssymb}
\usepackage{graphicx}
\usepackage{enumerate}

\theoremstyle{plain}
\newtheorem{claim}[theorem]{Claim}

\title{Time-dependent shortest paths in bounded treewidth graphs\footnote{This material is based upon work supported by the National Science Foundation under Grant No.\ CCF-1252833.}}

\author[1]{Glencora Borradaile}
\author[2]{Morgan Shirley}
\affil[1]{Oregon State University\\
  Corvallis, Oregon, US\\
  \texttt{glencora@eecs.oregonstate.edu}}
\affil[2]{Oregon State University\\
  Corvallis, Oregon, US\\
  \texttt{shirlemo@oregonstate.edu}}
\authorrunning{G. Borradaile and M. Shirley} 

\Copyright{Glencora Borradaile and Morgan Shirley}


\serieslogo{}
\volumeinfo
  {Billy Editor and Bill Editors}
  {2}
  {Conference title on which this volume is based on}
  {1}
  {1}
  {1}
\EventShortName{}
\DOI{10.4230/LIPIcs.xxx.yyy.p}

\begin{document}

\maketitle

\begin{abstract}
	We present a proof that the number of breakpoints in the arrival function between two terminals in graphs of treewidth $w$ is $n^{O(\log^2 w)}$ when the edge arrival functions are piecewise linear. This is an improvement on the bound of $n^{\Theta(\log n)}$ by Foschini, Hershberger, and Suri for graphs without any bound on treewidth. We provide an algorithm for calculating this arrival function using star-mesh transformations, a generalization of the wye-delta-wye transformations.
 \end{abstract}

	\section{Introduction}
	
	We consider the problem of computing shortest paths in a graph whose edge-costs are not constant, but depend on the time at which a traveler arrives at an endpoint.  This is used to model many real-world situations in which edge-costs are not fixed.  For example, in road networks, the cost of traversing a given segment of road depends on the time of day: travel times may be longer during rush hour.  Similarly, in routing networks, certain connections may experience more delay during peak downloading times.  For a given starting point $s$, a given starting time $t$, and positive edge-cost functions, one can modify Dijkstra's algorithm to account for the variable edge-costs by storing, in addition to a vertex's priority, the time at which one can arrive at that vertex (for details, see the algorithms of Orda and Rom~\cite{orda1990shortest} and Ding, Yu, and Qin~\cite{ding2008finding}).
	
	However, if we wish to compute the cost of traveling from one vertex to another as a function of \emph{all} possible departure times from the start vertex, the problem quickly becomes much more difficult.  
	Foschini, Hershberger, and Suri showed that, even for linear edge-cost functions, the number of times that a shortest path between two vertices can change, over the possible departure times, is in the worst case $n^{\Theta(\log n)}$~\cite{Foschini14}. A function representing the minimum cost of a shortest path between two vertices as a function of departure time is similarly bounded.
	
	Foschini, Hershberger, and Suri further observe that this bound should be polynomial for {\em bounded treewidth} graphs~\cite{Foschini14} based on a theorem by Fern{\'a}ndez-Baca and Slutzki~\cite{fernandez1992parametric} (we define bounded treewidth graphs in Section~\ref{sec:prelim}). In this paper, we give a constructive proof of this observation and an efficient algorithm for calculating the list of shortest paths between two vertices in bounded treewidth graphs. More specifically, we show that, given a graph of treewidth $w$ with linear edge-cost functions, the number of different shortest paths (over all possible departure times) between two vertices is bounded above by $n^{O(\log^2 w)}$ (Theorem~\ref{thm:bounds}). Given this bound, it is possible to bound the complexity of manipulating edge-cost functions algorithmically.  We provide an efficient method for calculating the list of shortest paths between two vertices in graphs of bounded treewidth (Theorem~\ref{thm:bounds_on_algorithm}). To do so, we use an algorithm for reducing the size of a graph (Theorem~\ref{thm:reduction_algorithm}) by way of parallel reductions and star-mesh transformations (defined in Section~\ref{sec:transformations}).
	
	\subsection{Related work}
	
	We briefly note a couple of papers dealing with time-dependent shortest paths that are not mentioned above. Work related to star-mesh transformations will be mentioned in Section~\ref{sec:transformations}, after defining these transformations formally.
	
	Cooke and Halsey \cite{cooke1966shortest} first introduced the idea of time-dependent shortest paths. They were concerned with finding the shortest path between any two vertices at a given time where the edges have \emph{discrete} timesteps, instead of the continuous range of times that we allow in this paper. 
	
	Dreyfus \cite{dreyfus1969appraisal} surveys a number of shortest paths problems, including time-dependent shortest paths. He references the work of Cooke and Halsey and improves it by allowing for continuous edge cost functions. 
	
	Dean \cite{dean2004shortest} provides a survey of work completed in this field. In the paper he notes that finding shortest paths at a specific time is much easier than finding shortest paths at all times -- a fact later given a strict bound by Foschini, Hershberger, and Suri~\cite{Foschini14}.
	
	\section{Preliminaries}\label{sec:prelim}
	
	\subsection{Treewidth}
	
	The following definitions are from Robertson and Seymour~\cite{robertson1984graph}.
	
	A \emph{tree decomposition} of a graph $G = (V, E)$ is a pair $(T, \mathcal{X})$ where $T = (V_T, E_T)$ is a tree and $\mathcal{X} = (X_t : t \in V_T)$ is a family of subsets of $V$ where the following hold:
	\begin{itemize}
		\item The union of all elements of $\mathcal{X}$ is $V$.
		\item For every edge $e \in E$ there exists $t \in V_T$ where $e$ has both ends in $X_t$.
		\item If $t, t', t'' \in V_T$ are in a path of $T$ in that order, then all vertices in the intersection of $X_t$ and $X_{t''}$ are also in $X_{t'}$.
	\end{itemize}
	Elements of $\mathcal{X}$ are called \emph{bags}. The \emph{width} of a tree decomposition is the maximum cardinality of bags in $\mathcal{X}$ minus one. The \emph{treewidth} of a graph is the minimum width over all tree decompositions. For instance, the treewidth of any tree graph is 1.
	
	\subsection{Time-dependent shortest paths}
	
	Consider a graph $G = (V, E)$. For each edge $uv \in E$, a trip along $uv$ departing from $u$ at time $t$ will arrive at $v$ at time $A_{uv}(t)$, where $A_{uv}: \mathbb{R}^{+} \cup \{\infty\} \rightarrow \mathbb{R}^{+} \cup \{\infty\}$.  We call $A_{uv}$ the \emph{arrival function} of $uv$.  Likewise, $A_{vu}$ is the arrival function for travel along $uv$ departing from $v$ and arriving at $u$.  If an edge is only traversable in one direction, the arrival function in the other direction is $\infty$. We use undirected edges here instead of directed edges in order to simplify the calculation of new edge arrival functions in Section~\ref{ssec:maintaining-arrival}.
	
	For all $uv$ and $t$, we require the following two constraints to ensure that the arrival function behaves reasonably. 
	
	\begin{itemize}
		\item $A_{uv}(t) \ge t$. That is, the traversal of an edge cannot be completed before it has begun.
		\item $\frac{d}{dt}(A_{uv}(t)) \ge 0$. This means that a later departure time cannot result in an earlier arrival time. Edges under this constraint are called \emph{First-In First-Out}, or \emph{FIFO}, because of the property that two traversals of an edge will complete in the order that they were initiated. This is the case for many applications of the time-dependent shortest paths problem.
	\end{itemize} 
	
	With these constraints, the set of arrival functions forms a semiring with the two operators relevant to this paper, $\min$ and $\circ$ (functional composition). With the requirement that $\frac{d}{dt}(A_{uv}(t)) \ge 0$, $\circ$ is left distributive over $\min$.
	
	The arrival function of a path $P = \{e_1, e_2, \ldots, e_{|P|}\}$, denoted $A_P$, is the composition of the arrival functions of all edges in that path. For two vertices $s$ and $s'$ and a given time $t$, $A_{(s, s')}(t)$ is the minimum value of $A_P(t)$ over all $s$-to-$s'$ paths $P$; the corresponding arrival function, $A_{(s,s')}$, for any two vertices $s$ and $s'$ is likewise defined. The arrival function between a vertex and itself is the identity function. In applications, we often want to find the arrival function for two distinguished vertices, called \emph{terminals}. We give the special name \emph{end-to-end arrival function} to $A_{(s, d)}(t)$ for terminals $s$ and $d$.
	
	When we are discussing correlated arrival functions in multiple graphs, we clarify with a superscript which graph the arrival function we are considering is in.  for example, $A_{P}^H$ is the arrival function for a path $P$ in a graph $H$ and $A_{(s, s')}^H$ is the arrival function for vertices $s$ and $s'$ in a graph $H$.
	
	\subsection{Graph transformations} \label{sec:transformations}
	
	In our algorithm for calculating end-to-end arrival functions we use \emph{parallel reductions} and \emph{star-mesh transformations}. Graph operations such as these have a wide range of uses, such as network analysis (for example, Chari, Feo, and Provan use such operations for approximating network reliability~\cite{chari1996delta}) and determining equivalent resistances in a circuit~\cite{kennelly1899equivalence}.
	
	\subsubsection{Wye-delta-wye transformations}
	
	One set of graph transformations that has received significant research attention is the set of \emph{wye-delta-wye} reductions, which include the series-parallel reductions along with two additional reductions, the Y-$\Delta$ and $\Delta$-Y reductions.  The series-parallel reductions are so-called because any series-parallel graph can be {\em reduced} (transformed by a sequence of these reductions to a single vertex) by repeated application of these steps. Series-parallel graphs are exactly the graphs with treewidth 2 (see, for example, Brandst{\"a}dt, Le, and Spinrad~\cite{brandstadt1999graph}). The addition of the Y-$\Delta$ and $\Delta$-Y reductions expand the set of all reducible graphs to include all planar graphs~\cite{Epifanov66}. 
	\paragraph{Series-parallel reductions}
	\begin{itemize}
		\item $R_{0}$: Delete a self-loop.
		\item $R_{1}$ (Pendant Reduction): Delete a degree-one vertex and its incident edge.
		\item $R_{2}$ (Series Reduction): Given a degree-two vertex $u$ adjacent to vertices $v$ and $w$, delete $u$ and replace the edges $uv$ and $uw$ with a single edge $vw$.
		\item $R_{3}$ (Parallel Reduction): Given a cycle of length two, delete one of the edges in the cycle.
	\end{itemize}
	\paragraph{Y-$\Delta$ and $\Delta$-Y transformations} A wye, Y, is a vertex of degree 3 and a delta, $\Delta$, is a cycle of length 3\footnote{In most cases, these transformations are applied to planar graphs, in which case $\Delta$s are usually restricted to be faces.}.
	
	\begin{itemize}
		\item Y-$\Delta$: Delete a wye $u$ with adjacent vertices $v$, $w$, and $x$ and replace edges $uv$, $uw$, and $ux$ with edges $vw$, $vx$, and $wx$.
		\item $\Delta$-Y: Delete a delta consisting of edges $vw$, $vx$, and $wx$ and add a vertex $u$ and edges $uv$, $uw$, and $ux$.
	\end{itemize}
	
	Note that whereas the series-parallel reductions each reduce the number of edges in a graph by one, the Y-$\Delta$ and $\Delta$-Y transformations keep the number of edges constant.  Also note that Y-$\Delta$ and $\Delta$-Y are reverse operations of each other. 
	
	We call two graphs \emph{wye-delta-wye equivalent} if it is possible through repeated application of the Y-$\Delta$ and $\Delta$-Y transformations to create one graph given the other. Naturally, this relationship is symmetric.
	
	It is often of interest to indicate a set of terminal vertices that should remain at the end of a series of reductions:  the terminals should not be deleted as part of an $R_{1}$, $R_{2}$, or Y-$\Delta$ transformation.
	
	\subsubsection{Related work: wye-delta-wye reducibility}
	
	Epifanov~\cite{Epifanov66} was the first to prove that all planar graphs are wye-delta-wye reducible. Feo and Provan~\cite{Feo93} give a simple algorithm for reducing two-terminal planar graphs using $O(n^{2})$ transformations. Chang and Erickson~\cite{chang2015electrical} prove that there are some graphs for which $\Omega(n^{3/2})$ transformations are required. Both these papers conjecture that there exists an algorithm for wye-delta-wye reduction of planar graphs using $\Theta(n^{3/2})$ transformations. Gitler and Sagols~\cite{Gitler11} give a $O(n^{4})$ algorithm for reducing three-terminal planar graphs; Archdeacon, Colbourn, Gitler, and Provan~\cite{Archdeacon00} show that the existence of such an algorithm implies that one-terminal crossing-number-one graphs are also reducible.
	
	There is no known characterization of wye-delta-wye reducible graphs, but since it is a minor-closed family~\cite{truemper1989delta}, the Robertson-Seymour theorem guarantees the existence of a finite number of forbidden minors~\cite{robertson2004graph}, each of which can be recognized in polynomial time~\cite{robertson1995graph}. Yu~\cite{Yu06} gives a proof that there are more than 68 billion such forbidden minors, so while recognizing wye-delta-wye reducible graph is in $\textsf{P}$, an algorithm relying on detecting forbidden minors would be impractical.
	Seven known forbidden minors are the Petersen Family of graphs. These graphs include the Petersen Graph and its 6 wye-delta-wye equivalent graphs (including $K_{6}$ and $K_{3,3,1}$). Since these are the 7 forbidden minors for linklessly-embeddable graphs~\cite{robertson1995sachs}, all wye-delta-wye reducible graphs are linklessly-embeddable.
	
	Some graphs are not reducible to a single vertex but are reducible to a smaller irreducible graph. For example, it is easy to show that the Heawood graph reduces to $K_{7}$ and the M\"obius-Kantor graph reduces to $K_{2,2,2,2}$. Other graphs, however, cannot have any of the wye-delta-wye transformations applied to them because they have both minimum degree and girth 4; for example, the four-dimensional hypercube graph $Q_{4}$.
	
	\subsubsection{Star-mesh transformations}
	
	A natural generalization of the serial reduction and the Y-$\Delta$ transformation is to increase the size of the deleted vertex and of the resulting clique. This general class of transformations are called \emph{star-mesh transformations}. We call such a transformation for a deleted vertex of degree $k$ a $k$-star-mesh transformation. Note that the size of the resulting clique (the ``mesh'') is $k$ as well. This class of transformations is well-studied, especially in its application to electrical networks; see, for example, the work of Shannon~\cite{shannon1938symbolic}.
	
	One might consider an inverse transformation where the edges of a clique are deleted and a star is added adjacent to all vertices previously in the clique. This is a natural generalization of the $\Delta$-Y transformation. However, a $k$-star-mesh transformation for $k > 3$ will increase the number of edges in the graph. This means that the equivalent $k$-mesh-star transformation will \emph{reduce} the number of edges in the graph. If edges in the graph are assigned weights, and we are expecting some property of the graph to be maintained after the transformation, we will assign the new edges weights based on some set of equations; a reduction in the number of edges can result in there being fewer variables than equations, leading to an unsolvable system.
	
	\begin{figure}
		\centering
		
		\caption{A 4-star-mesh transformation.}
		\vspace*{2mm}
		\includegraphics[width=.4\linewidth]{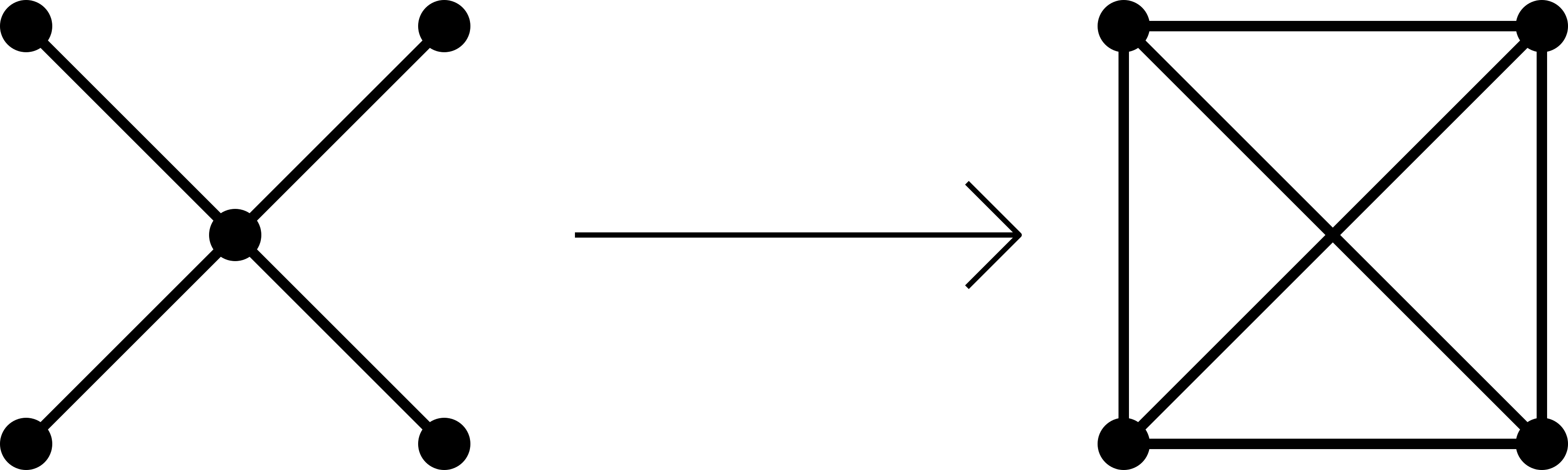}
	\end{figure}
	
	Any graph can be trivially reduced by way of star-mesh transformations of arbitrary size. One can simply choose a vertex and star-mesh transform it. The resulting graph has strictly fewer vertices, and therefore this process can be continued until only one vertex is remaining. However, this can result in a very dense graph; one goal of this paper is to maintain sparsity.
	
	\section{Polynomial bounds on arrival functions}
	We consider time-dependent shortest paths in graphs with continuous piecewise-linear edge arrival functions. In a piecewise-linear function $f: \mathbb{R} \rightarrow \mathbb{R}$, a \emph{breakpoint} is a value $t$ where $\exists \alpha: \forall \varepsilon \mbox{ with } 0 < \varepsilon < \alpha, f'(t-\varepsilon) \ne f'(t+\varepsilon)$. In simpler terms, a breakpoint is a point at which one ``piece'' of the function ends and another begins. 
	
	When manipulating such a graph we calculate new arrival functions (for paths and pairs of vertices) as minima and compositions of other arrival functions. When storing and performing computations on piecewise-linear functions, the complexity of operations depends on the number of breakpoints. The number of breakpoints that can result in an end-to-end arrival function gives a lower bound on the complexity of computing time-dependent shortest paths over all times in graphs with piecewise-linear edge arrival functions.
	
	Let $b^G$ be the maximum number of breakpoints in any end-to-end arrival function of a graph $G$ ($A_{(s,d)}^G$ for any vertices $s,d \in G$). We will use the notation $b_w^G$ when $G$ has treewidth $w$. Similarly, we let $b(n)$ be the maximum number of breakpoints in any end-to-end arrival function for \emph{any} graph with up to $n$ vertices, and $b_w(n)$ for any graph with up to $n$ vertices and treewidth at most $w$. Finally, if $s$ and $d$ are vertices in $G$, then $b^G_{(s, d)}$ is the number of breakpoints in the $s$-to-$d$ arrival function.
	
	Foschini, Hershberger, and Suri~\cite{Foschini14} prove that $b(n) = Kn^{\Theta(\log n)}$, where $K$ is the total number of linear pieces among all the edge arrival functions in the initial graph; that is $K$ is at most the number of primitive breakpoints plus the number of edges.  In this paper, we prove that $b_w(n) = Kn^{O(\log^2 w)}$. In our proof we use the following two results of Foschini, Hershberger, and Suri; we have reworded their statements to be consistent with the notation of this paper.
	
	
	\begin{lemma}[Lemma 4.2, Foschini, Hershberger, and Suri~\cite{Foschini14}] \label{lem:k_times_more}
		The number of breakpoints in an end-to-end arrival function in a graph with piecewise linear edge arrival functions is at most $K$ times the number of breakpoints in the same function if the graph had linear edge functions. That is, $b^G_{(s,d)} \le K \cdot b^{G'}_{(s,d)}$ where $G$ has $K$ linear pieces among all the edge arrival functions and $G'$ has linear edge arrival functions, for any terminal vertices $s, d$.
	\end{lemma}	
	
	\begin{theorem}[Theorem 4.4, Foschini, Hershberger, and Suri~\cite{Foschini14}] \label{thm:n_to_the_log_n}
		For any graph with $n$ nodes and linear edge arrival time functions, the number of breakpoints is at most $n^{O(\log n)}$.  That is, $b(n) = n^{O(\log n)}$.
	\end{theorem}
	
	We use these results to prove the following stronger bound for graphs of bounded treewidth by induction over a given tree decomposition.  We use Theorem~\ref{thm:n_to_the_log_n} in the base case of the induction and use Lemma~\ref{lem:k_times_more} in the inductive step.
	
	\begin{theorem} \label{thm:bounds}
		The maximum number of breakpoints in an end-to-end arrival function for a graph $G$ of treewidth $w$ with $n$ vertices and piecewise linear edge arrival functions with at most $K$ pieces in the entire graph is at most $Kn^{O(\log^2 w)}$. That is, $b_w(n) = Kn^{O(\log^2 w)}$.
	\end{theorem}
	\begin{proof}
		Consider a graph of treewidth $w$ with $n_0 \le 2w + 2$ vertices. From Theorem~\ref{thm:n_to_the_log_n} we know that 
		\begin{equation}
		b_w(n_0) = (2w + 2)^{O(\log(2w+2))} = w^{O(\log(w))}.\label{eq:basecase}
		\end{equation}
		
		It is well known that for any graph of treewidth $w$ where $n > 2w + 2$, there exists a separator $S$ of size at most $w + 1$ that divides the graph into two subsets $V_1, V_2$, each of which contains at most $\frac{2n}{3}$ vertices. Let $V' = S \cup \{s, d\}$ where $s$ and $d$ are the terminal vertices. Note that $|V'| \leq w + 3$.

		We construct a graph $G'$ on the vertex set $V'$ with one or two edges between every pair of vertices with assigned edge arrival functions derived from arrival functions in induced subgraphs of $G$.  First, consider the induced graph $G[V_1 \cup S]$, that is, vertices on one side of and including $S$. For every $u, v \in S \cup (\{s, d\} \cap V_1)$, add an edge $uv$ to $G'$ with arrival function $A^{G[V_1 \cup S]}_{(u, v)}$ (and $A^{G[V_1 \cup S]}_{(v, u)}$ for the reverse direction). Note that it is possible that $v$ is not reachable from $u$ in $G[V_1 \cup S]$; in this case, $A^{G[V_1 \cup S]}_{(u, v)} = \infty$. This edge then represents the time necessary to travel between $u$ and $v$ in $G$ only using edges on one side of $S$.  Second, add additional edges from the induced graph $G[V_2 \cup S]$ in the same way.  Let $E'$ be the resulting set of edges.  Note that for $u,v \in S$ there are parallel edges between $u$ and $v$ in $E'$, but if, for example, $s \notin S$, then edges incident to $s$ will not have parallel counterparts.  In this way, edges in $E'$ correspond to paths in $G$ between vertices in $V'$ that only contain edges on one side or other of the separator.
		
		\begin{claim}\label{clm:same}
			$A_{(s, d)}^{G'} = A_{(s, d)}^{G}$.  In particular, these functions have the same number of breakpoints.
		\end{claim}
		
		\begin{proof}[Proof of Claim~\ref{clm:same}]
			Consider an arbitrary departure time $t$.  Let $P_t$ (respectively~$P_t'$) be the shortest path to $d$ departing from $s$ at time $t$ in graph $G$ (respectively~$G'$). We argue that the time to traverse $P_t$ equals the time to traverse $P_t'$, that is, that $A_{P_t}^G(t) = A_{P_t'}^{G'}(t)$, proving the claim.
			
			The path $P_t$ corresponds to a path of equal length in $G'$. 
			This is clear because any shortest path will either not go
			through $S$, in which case an edge corresponding to it will be in
			$G'$ by construction, or will go through some vertices $v_1, \ldots,
			v_k \in S$, in which case $P_t = A_{(s, d)}^G(t) = A_{(v_k, d)}^G(t) \circ
			\ldots \circ A_{(s, v_1)}^G(t)$. All of the latter paths have edges
			corresponding to them in $G'$ by construction.
			
			The path $P'_t$ corresponds to a path of equal length in $G$. 		
			Consider the case where $P'_t$ does not go
			through $S$. Then the edge $sd \in E'$ has arrival time function
			$A_{sd}^{G'}(t) = A_{(s, d)}^G$ by construction. If $P'_t$ does go 
			through some vertices $v_1, \ldots, v_k \in S$, then
			there is some walk that is the concatenation of shortest paths
			between $s$ and $v_1$, $v_i$ and $v_{i+1}$, and $v_k$ and $d$ in
			$G$, again by the construction of $G'$. To show that this walk in
			$G$ is indeed a path, consider intermediate paths from $v_a$ to
			$v_b$ and $v_c$ to $v_d$. If these paths share any vertex $v_j$,
			then because all edges $e$ in $G$ have the property that $A_{e}(t)
			\ge t$ we could replace these paths (and all paths between them in
			the walk) with the paths from $v_a$ to $v_j$ and $v_j$ to $v_d$ to
			get a walk that is shorter than or equal to our original walk. If it is equal in length, we can let $P'_{t}$ correspond to this new walk instead, as it is the same length and visits vertex $v_j$ at least one fewer time than before. Then we can repeat this process until no vertex is visited more than once. If it is shorter, however, we arrive at a contradiction because we said that $P'_{t}$ was the shortest path between $s$ and $d$ in $G'$, and this shortcut from $v_a$ to $v_d$ would by construction of $G'$ imply that there is a shorter path in $G'$ that bypasses $v_b$ and $v_c$,
			which leads to the conclusion that there is no such shared vertex
			$v_j$.
			
			Therefore, a path $Q'_t$ of the same length as $P_t$ exists in $G'$ between $s$ and $d$ and a path $Q_t$ of the same length as $P'_t$ exists in $G$ between $s$ and $d$. Since $P_t$ and $P'_t$ are the shortest paths between $s$ and $d$ in their respective graphs, $P_t$ is no longer than $Q_t$ and $P'_t$ is no longer than $Q'_t$. Therefore, all of these paths have the same length. This completes the proof of Claim~\ref{clm:same}.
		\end{proof}

		Each edge of $E'$ represents a trip between vertices of $V'$ in $G$ that visits at most $2n/3$ vertices; therefore, each edge of $E'$ has an arrival function with at most $b_w(2n/3)$ breakpoints. Since there are $O(w^2)$ edges in $E'$, $G'$ has a total of $O(w^2 \cdot b_w(2n/3))$ breakpoints (and linear segments). If $E'$ had linear edge arrival functions, then by Equation~(\ref{eq:basecase}), there would be 
		$w^{O(\log w)}$ breakpoints in end-to-end arrival functions of $G'$. By Lemma~\ref{lem:k_times_more}, the number of breakpoints in end-to-end arrival functions of $G'$ is therefore $w^{O{(\log w)}} \cdot O(w^2) b_w(2n/3)$.  Since the arrival functions in $G$ and $G'$ are equal (Claim~\ref{clm:same}), the number of breakpoints in $G$ is described by the following recurrence:
		\[
		b_w(n) = w^{O(\log w)} b_w(2n/3) 
		\]
		
		Solving this recurrence with the base case given in Equation~(\ref{eq:basecase}), we get that $b_w(n) = n^{O(\log^2 w)}$, assuming that $G$ has linear edge arrival functions.  Invoking Lemma~\ref{lem:k_times_more} completes the proof of Theorem~\ref{thm:bounds}.
	\end{proof}
	
	\section{Time-dependent shortest paths in graphs of bounded treewidth}
	
	In this section we describe a method for reducing a graph of bounded treewidth to a single edge between two terminal vertices, using series-parallel reductions as well as star-mesh transformations, while maintaining the end-to-end arrival function between these two vertices. First, in Lemma~\ref{lem:transformations_work}, we will give a method for reassigning the edge arrival functions of a graph during each of the relevant transformations that will preserve arrival functions between the remaining vertices of the graph. Second, in Theorem~\ref{thm:reduction_algorithm} we will show that graphs of bounded treewidth and two terminals can be efficiently reduced using only these transformations, with a bound on the degree of the deleted vertex that only exceeds the treewidth of the graph by one. Finally, we will show that this result, together with Theorem~\ref{thm:bounds}, implies that end-to-end arrival functions can be efficiently computed in graphs of bounded treewidth. This is Theorem~\ref{thm:bounds_on_algorithm}, the main result of the paper.
	
	In the following discussion we will differentiate parallel edges with a subscript.  That is, if there are $k$ edges between vertices $u$ and $v$, then we will denote the edges as $(uv)_1, (uv)_2, \ldots, (uv)_k$.
	
	\subsection{Maintaining arrival functions} \label{ssec:maintaining-arrival}
	
	We start by showing that we can correctly maintain arrival functions under star-mesh transformations.  For obvious reasons, we don't allow the terminal vertices ($s$ and $d$) to be deleted in such transformation.  Self-loop deletions and pendant reductions (not involving terminals) clearly do not affect end-to-end arrival functions, given our realistic constraints on arrival functions.   For the parallel reduction, in which parallel edges $(uv)_1$ and $(uv)_2$ are replaced with a single edge $uv$, we set 
	\begin{equation}
	A_{uv}(t) = \min\{A_{(uv)_1}(t), A_{(uv)_2}(t)\} \mbox{ and }A_{vu}(t) = \min\{A_{(vu)_1}(t), A_{(vu)_2}(t)\}.  \label{eq:parallel}
	\end{equation}
	In the star-mesh transformation (and, as a special case, the series reduction), a vertex $c$, with neighbors $v_1, v_2, \ldots, v_d$, is deleted and edges $v_iv_j$ for all $i < j$ are added.  For each edge $v_i v_j$ in the resulting graph, we set 
	\begin{equation}
	A_{v_i v_j}(t) = A_{c v_j} \circ A_{v_i c}(t) \mbox{ and }A_{v_j v_i}(t) = A_{c v_i} \circ A_{v_j c}(t).\label{eq:startomesh}
	\end{equation}
	\begin{lemma} \label{lem:transformations_work}
		Parallel reductions and star-mesh transformations  (Equations~(\ref{eq:parallel}) and~(\ref{eq:startomesh})) preserve end-to-end arrival functions.
	\end{lemma}
	\begin{proof}
		Since composition of functions is associative, any path arrival function can be written as compositions of the arrival functions of segments of that path. This means that for any fixed $t$, $A_{(s, d)}(t) = A_{(u, d)} \circ A_{(s, u)}(t)$ for any vertex $u$ on an $s$-to-$d$ path $P$ for which $A_P(t) = A_{(s, d)}(t)$.
		
		Consider the parallel reduction of $(uv)_1,(uv)_2$ to $uv$.  We argue that $A_{(u,v)}(t)$ is preserved by the assignment of Equation~(\ref{eq:parallel}) for all departure times $t$;  since  $A_{(u, v)}(t)$ is not changed for any fixed $t$, the value of $A_{(s, s')}(t)$ is also preserved.  Let $P_{(u,v)}(t)$ be the shortest $u$-to-$v$ path departing from $u$ at time $t$.  There are two cases: neither $(uv)_1$ nor $(uv)_2$ is in $P_{(u,v)}(t)$ or one of $(uv)_1$ and $(uv)_2$ (w.l.o.g., say $(uv)_1$) is in $P_{(u,v)}(t)$.  In the first case, we can safely ignore the parallel reduction because the added edge ($uv$) will not have an arrival function with value less than the removed edges. The second case, $A_{(uv)_1}(t)$ is the minimum arrival time at $v$, departing from $u$ at time $t$; by Equation~(\ref{eq:parallel}), $A_{uv}(t) = A_{(uv)_1}(t)$ as required.
		
		The case of star-mesh transformations follows a similar line of reasoning.   Let $P_{(s,d)}(t)$ be the shortest $s$-to-$d$ path departing $s$ at time $t$.  There are two cases: $c \notin P_{(s,d)}(t)$ and $c \in P_{(s,d)}(t)$.  In the first case, the star-mesh transformation does not impact $A_{(s,d)}(t)$ because the added edges will not have arrival functions with values less than alternate routes.  In the second case, let $uc$ and $cv$ be the edges incident to star vertex $c$ in $P_{(s,d)}(t)$.  Then $A_{(s,d)}(t) = A_{(v, d)}\circ A_{cv} \circ A_{uc} \circ A_{(s, u)}(t)$ by the associativity of composition.  By Equation~(\ref{eq:startomesh}), $A_{(v, d)}\circ A_{cv} \circ A_{uc} \circ A_{(s, u)}(t) = A_{(v, d)}\circ A_{uv} \circ A_{(s, u)}(t)$, as desired.   Further, any other path visiting a pair of vertices, say $x$ and $y$ incident to $c$ will have no less an arrival time; that is, $A_{(s,d)}(t) \ge A_{(y,d)} \circ A_{(x,y)} \circ A_{(s,x)}(t) \ge A_{(y,d)} \circ A_{xy} \circ A_{(s,x)}(t)$ where the second inequality follows from Equation~(\ref{eq:startomesh}).
	\end{proof}
	
	\subsection{Efficiently reducing graphs of bounded treewidth}
	
	Using these transformations we can find $A_{(s, d)}(t)$ for any graph by repeatedly applying star-mesh transformations and parallel reductions until the graph is the single edge $sd$.  At this point, by  Lemma~\ref{lem:transformations_work}, $A_{(s,d)} = A_{sd}$. Unfortunately there are two significant drawbacks to this method.
	The first problem is that a star-mesh transformation on a vertex of degree $d$ creates $\binom{d}{2}$ new edges, requiring $O(d^2)$ composition operations. In general graphs this can be as bad as $O(n^2)$ new edges. The entire reduction, which does $|V|-2$ star-mesh transformations, will take $O(n^3)$ composition operations to complete. The second problem is that the edge arrival functions can themselves gain too many linear segments, no longer supporting efficient calculations, as evidenced by Theorem~\ref{thm:n_to_the_log_n}.
	
	For graphs of bounded treewidth, the second problem is solved by the polynomial bound on the number of breakpoints in end-to-end arrival functions provided by Theorem~\ref{thm:bounds}.  We improve on the limitations suggested by the first problem by showing that graphs of bounded treewidth with 2 terminals can be reduced with a linear number of star-mesh transformations on stars of size depending only on the treewidth.  This generalizes El-Mallah and Colbourn's~\cite{ElMallah90} result that all graphs of treewidth 3 without terminals can be wye-delta reduced (i.e.\ star-mesh reduced with stars of degree at most 3).  Formally, we show:
	\begin{theorem} \label{thm:reduction_algorithm}
		A two-terminal graph $G$ with $n$ vertices and treewidth at most $w$ can be reduced using $O(w^2 n)$ parallel reductions and $O(n)$ star-mesh transformations of degree at most $w+1$.
	\end{theorem}
	
	To simplify the presentation of our proofs, we use {\em nice} tree decompositions.
	A \emph{nice tree decomposition} $(T, \mathcal{X})$ of a graph is a tree decomposition such that (in addition to properties we do not require for this paper) for any adjacent bags $X_{i}$ and $X_{j}$ in $\mathcal{X}$ either $X_{i} = X_{j}$, $X_{i} = X_{j} \cup \{v\}$, or $X_{j} = X_{i} \cup \{v\}$ for some vertex $v$. Tree decompositions can be made nice in linear time~\cite{Kloks94}; for a graph with $n$ vertices there is a nice tree decomposition with $O(n)$ bags.
	Given a nice tree decomposition $T$ of a treewidth $w$ graph $G$, we make the following assumptions at each step of the reduction process:
	
	\begin{enumerate}[{A}1]
		\item There are no parallel edges in $G$. If such edges exist we can simply parallel reduce them.  For every star-mesh transformation of degree $k$, at most $\binom{k}{2}$ parallel edges are introduced.  Therefore, if we perform $\ell$ star-mesh transformations of degree at most $w+1$, at most $O(w^2\ell)$ parallel reductions will be required.
		\item For all leaf bags $X_i \in T$ with parent $X_j$, $X_i \supsetneq X_j$. If this is not the case, then $X_i \subseteq X_j$ which means that we can safely remove $X_i$ from $T$ while maintaining the nice tree decomposition property of $T$.
		\item \label{assumption3} There is more than one bag in $T$. If there is only one, there are $w+1$ or fewer vertices remaining, each of which has maximum degree $w$. We can simply star-mesh transform each of the non-terminal vertices (performing parallel reductions as applicable) until only terminals remain, at which point the reduction is complete.
	\end{enumerate}
	
	Note that, due to A1, the degree of a vertex is the same as the number of vertices it is adjacent to (that is, we can ignore parallel edges). Therefore, we will refer to these values interchangeably when operating under these assumptions.
	
	At a high level, we reduce the graph by repeated elimination of leaf bags of $T$ that do not contain terminals until we are left with a path of bags, and then transform that path until we are left with a single bag that we can reduce as described in A3.
	
	\begin{lemma} \label{lem:path}
		Given a nice tree decomposition $T$ of width $w$ 
		of a graph with two terminals and where the above assumptions hold, either $T$ is a path or there is a leaf-bag $X_i$ with parent $X_j$ such that $X_i \setminus X_j$ is not a terminal.
	\end{lemma}
	
	\begin{proof}
		Assume that $T$ is not a path. Since $T$ is a tree, there must then be three or more leaf bags. By A2, a leaf bag $X_i$ is a strict superset of its parent $X_j$. This means that $X_i \setminus X_j$ is non-empty. Specifically, it is a single vertex that only appears in $X_i$. Because there are two terminals, only two of these vertices that are exclusive to a single leaf bag can be terminals. However, since there are at least three leaf bags, at least one leaf bag must contain a vertex exclusive to that bag that is not a terminal.
	\end{proof}
	
	\begin{lemma} \label{lem:kplusone_transform}
		Given a nice tree decomposition $T$ of width $w$ 
		of a graph $G$ with two terminals and where the above assumptions hold, there is a non-terminal vertex $v$ that can be removed by way of a $(w+1)$-star-mesh transformation without increasing the treewidth of $G$.
	\end{lemma}
	\begin{proof}
		
		If there exists a leaf-bag $X_i$ in $G$ with parent $X_j$ such that $X_i \setminus X_j$ is not a terminal, let $v = X_i \setminus X_j$. Because $T$ is a valid tree decomposition, $v$ can only be adjacent to the other vertices in $X_i$, of which there are at most $w$. Therefore, $v$ can be removed by way of a $w$-star-mesh transformation. Since the elements of $X_i$ were the only vertices affected by this transformation, any added edges have both endpoints inside $X_i$, leaving the validity of the tree decomposition $T$ unaffected.
		
		If no such leaf-bag exists, by Lemma~\ref{lem:path}, $T$ is a path, which we root arbitrarily at an endpoint of the path:  label the bags in order $X_1, X_2, \ldots, X_{|T|}$ where $X_{|T|}$ is the root bag.  Due to A2, we know that $X_1 \not \subseteq X_2$. Additionally, since the tree decomposition is nice, $X_1$ and $X_2$ differ by a single vertex $x$.  If $x$ is not a terminal we would be able to remove $x$ by way of a star-mesh transformation as described above.   Therefore, we assume that $x$ is terminal $s$, without loss of generality.
		
		Let $j$ be the lowest index with $j > 1$ such that $X_j \supset X_{j+1}$. If there is no such index, let $j = |T|$, that is, let $X_j$ be the root bag of the path. Clearly, then, $X_2 \subseteq \ldots \subseteq X_j$. 
		
		If $X_j$ is not the root bag, then we can choose $v$ to be the vertex in $X_j \setminus X_{j+1}$. In this case, $v$ may be adjacent to $X_j \cup \{s\}$, as it may be present in any bag with index less than $j$, but cannot be adjacent to any other vertex, as $v \not\in X_{j+1}$ and $T$ is a valid tree decomposition. Therefore, the number of adjacencies that $v$ has is $|X_j \cup \{s\} \setminus \{v\} | \leq w + 1$.
		
		If $X_j$ is the root bag, we can choose $v$ to be any non-terminal vertex in $X_j$. Because $X_j$ has size at most $w + 1$ and includes every vertex in the graph besides $s$, there are at most $w + 2$ vertices in the graph, out of which $w + 1$ are not $v$. Therefore, $v$ can only be adjacent to at most $w + 1$ other vertices.
		
		The chosen vertex $v$ is clearly degree $w + 1$, which implies that it can be removed using a $(w + 1)$-star-mesh transformation. It remains to show that the resulting graph still has treewidth $w$.
		
		If $X_j$ was the root bag before the deletion, we have shown that the graph had at most $w + 2$ vertices. With one of those vertices deleted, there are now $w + 1$ vertices. Since these will all fit into a single bag in a tree decomposition of width $w$, clearly the graph still has treewidth $w$. 	
		If $X_j$ was not the root bag, then a similar argument applies. The subgraph $X_j \cup \{s\}$, which is the portion of $G$ affected by the star-mesh transformation, had at most $w + 2$ vertices. With one of these vertices deleted, there are now $w + 1$ vertices. In addition, with the removal of $v$ there is now only one vertex, $s$, that this subgraph does not share with $X_{j+1}$. Therefore, we can combine this entire subgraph into a single bag with parent $X_{j+1}$, maintaining a nice tree decomposition with width $w$.
	\end{proof}
	
	Theorem~\ref{thm:reduction_algorithm} follows from Lemma~\ref{lem:kplusone_transform} and the assumptions. We perform $(w + 1)$-star-mesh transformations on the graph as described in Lemma~\ref{lem:kplusone_transform} until the graph has only two terminals remaining. There will be exactly $n-2$ of these transformations, as we can remove every vertex except for the terminals. Between these transformations, we reduce every set of parallel edges in $G$ as described by A1. Because there are $n-2$ vertices that are star-mesh transformed, each of which has degree at most $w + 1$, the number of parallel reductions will be at most $(n - 2){\binom{w + 1}{2}}$, which is $O(w^2n)$.
	
	A simple algorithm for solving time-dependent shortest paths in bounded treewidth graphs immediately follows. Simply reduce the graph as described above, maintaining the end-to-end arrival function between the terminals as in Lemma~\ref{lem:transformations_work}. Then the final edge, with its endpoints as the two terminals, will have the desired arrival function.
	
	\begin{theorem} \label{thm:bounds_on_algorithm}
		End-to-end arrival functions in a graph $G$ with treewidth $w$ can be computed in $n^{O(\log^2 w)}$ time.
	\end{theorem}
	\begin{proof}
		By Theorem~\ref{thm:reduction_algorithm}, we can reduce $G$ using $O(w^2 n)$ parallel reductions and $O(n)$ star-mesh transformations of maximum degree $w+1$. 
		
		For each parallel reduction we compute the minimum of two piecewise linear functions, which can be done in time linear in the number of breakpoints of the functions. One can simply iterate through the linear pieces of each function, noting when the functions intersect.
		
		Similarly, for each $k$-star-mesh transformation we compute $2{\binom{k}{2}} \in O(k^2)$ compositions of piecewise linear functions, one for each direction of each new edge, each of which can also be done time linear in the number of breakpoints of the functions. To compute the composition of functions $g \circ f$, one computes the image in $g$ of breakpoints of $f$, which is guaranteed to be sorted because the functions are monotone, and merges the image with a list of breakpoints of $g$. Then, for each interval in the merged list of breakpoints one calculates the value of the composition of the two relevant segments of the original functions using simple algebra. Since $k \leq w+1$, the number of compositions of performed is $O(w^2)$. 
		
		By Theorem~\ref{thm:bounds}, we know that each end-to-end arrival function in $G$ has $n^{O(\log^2 w)}$ breakpoints, which means that every edge at any stage of the reduction is similarly bounded. Therefore, the process will take $O(w^2 n)\cdot n^{O(\log^2 w)}$ time for the parallel reductions and $O(n) \cdot O(w^2) \cdot n^{O(\log^2 w)}$ time for the star-mesh transformations, for a total of $w^2 n^{O(\log^2 w)} = n^{O(\log^2 w)}$ time.
	\end{proof}
	
	\section{Future work}
	
	In this paper we showed that extending the wye-delta-wye transformations by adding star-mesh transformations of bounded degree greater than three allows for the efficient reduction of graphs of bounded treewidth. A natural question is if other classes of graphs can be efficiently reduced with similar extensions. In our algorithm for Theorem~\ref{thm:reduction_algorithm}, we never use the $\Delta$-Y transformation. Does using higher-degree star-mesh transformations in conjunction with the $\Delta$-Y transformation yield a reduction algorithm for an interesting set of graphs?
	
	The bound given by Foschini, Hershberger, and Suri~\cite{Foschini14} for general graphs is tight: that is, there exists a graph for which an end-to-end arrival function has $n^{\Theta(\log n)}$ breakpoints, and there are no graphs where any end-to-end arrival function is asymptotically worse than this. Their lower bound proof method extends to graphs of bounded treewidth. They construct a layered graph for which the layers have a size dependent on $n$. These layers have the property of being valid bags for a tree decomposition of the graph, so we can instead restrict the layers to have maximum size $w + 1$ for some constant width $w$ to get a bound of $n^{\Omega(\log w)}$.  However, 
	Theorem~\ref{thm:bounds} gives an upper bound of $n^{O(\log^2 w)}$. It remains open what a tight bound would be for bounded-treewidth graphs.

\pagebreak

\bibliography{bibliography}

\end{document}